\documentclass[a4paper,12pt]{article}
\usepackage{bbding}
\usepackage{bbm}
\usepackage{amsmath}
\usepackage{amsthm}
\newtheorem{assump}{Assumption}
\newtheorem{thm}{Theorem}
\usepackage{bm}
\usepackage{appendix}
\usepackage{graphicx}
\usepackage{natbib}
\usepackage{hyperref}
\DeclareMathOperator*{\argmin}{arg\,min}

\newcommand{\bbeta}{{\boldsymbol{\beta}}}
\newcommand{\btheta}{{\boldsymbol{\theta}}}

\newcommand{\bTheta}{{\boldsymbol{\Theta}}}

\newcommand{\bX}{\mathbf{X}}
\newcommand{\bx}{\mathbf{x}}
\newcommand{\bb}{\mathbf{b}}
\newcommand{\calR}{\mathcal{R}}
\newcommand{\R}{\mathbbm{R}}
\newcommand{\Sp}{\mathbbm{S}}
\newcommand{\KL}{\mathbbm{KL}}

\begin{document}
\title{Regularized Maximum Likelihood Estimation for the Random Coefficients Model}
\author{\begin{tabular}{ccc}
Fabian Dunker\footnote{School of Mathematics and Statistics, University of Canterbury, Private Bag 4800, Christchurch 8140, New Zealand, fabian.dunker@canterbury.ac.nz} & Emil Mendoza\footnote{School of Mathematics and Statistics, University of Canterbury, Private Bag 4800, Christchurch 8140, New Zealand, emil.mendoza@pg.canterbury.ac.nz} \footnote{Corresponding author} & Marco Reale\footnote{School of Mathematics and Statistics, University of Canterbury, Private Bag 4800, Christchurch 8140, New Zealand, marco.reale@canterbury.ac.nz}\\
\small{University of Canterbury}  & \small{University of Canterbury}  & \small{University of Canterbury}
\end{tabular}
}
\maketitle

\begin{abstract}
\noindent 
The random coefficients model $Y_i={\beta_0}_i+{\beta_1}_i {X_1}_i+{\beta_2}_i {X_2}_i+\ldots+{\beta_d}_i {X_d}_i$, with $\bX_i$, $Y_i$, $\bbeta_i$ i.i.d, and $\bbeta_i$ independent of $\bX_i$ is often used to capture unobserved heterogeneity in a population.  We propose a quasi-maximum likelihood method to estimate the joint density distribution of the random coefficient model. This method implicitly involves the inversion of the Radon transformation in order to reconstruct the joint distribution, and hence is an inverse problem. Nonparametric estimation for the joint density of $\bbeta_i=({\beta_0}_i,\ldots, {\beta_d}_i)$ based on kernel methods or Fourier inversion have been proposed in recent years. Most of these methods assume a heavy tailed design density $f_\bX$. To add stability to the solution, we apply regularization methods. We analyze the convergence of the method without assuming heavy tails for $f_\bX$ and illustrate performance by applying the method on simulated and real data. To add stability to the solution, we apply a Tikhonov-type regularization method.
\end{abstract}

\section{Introduction}\label{sec:intro}

Unobserved heterogeneity in a population is an important problem in statistics and econometrics. A common way to account for unobserved heterogeneity is the use of random effects or mixed models which usually impose parametric assumptions on the distribution of the random effects. This is a significant restriction on the structure of the heterogeneity. In recent years the linear random coefficient model, which is a nonparametric alternative, attracted significant attention in the literature. This model is given by
\begin{equation}\label{eqn:rc_model}
	Y_i=\beta_{0i}+\beta_{1i} X_{1i}+\beta_{2i} X_{2i}+\ldots+\beta_{di} X_{di}.
\end{equation}
With $\bbeta_i = (\beta_{0i}, \beta_{1i}, \beta_{2i}, \ldots, \beta_{di})^\top$ and $\bX_i = (1, X_{1i}, X_{2i}, \ldots X_{di})^\top$ we can write \eqref{eqn:rc_model} as $Y_i = \bbeta_i^\top \bX_i$. We assume that $\bX_i$, $Y_i$, $\bbeta_i$ are i.i.d random variables with $\bbeta_i$ and $\bX_i$ being independent, and that $Y_i$ and $\bX_i$ are observed while $\bbeta_i$ is unknown. The objective is to estimate the joint density $f_\bbeta$ of $\bbeta_i$ nonparametrically. 

In this paper we propose and analyze a regularized maximum likelihood estimator for the density $f_\bbeta$
\begin{align}\label{eqn:method1}
\hat f_\bbeta = \argmin_{f\ge 0,\, \|f\|_{L^1}=1} -\ell(f|Y_1,\bX_1, \ldots Y_n,\bX_n)+\alpha \calR(f)
\end{align}
Here $\ell$ is the log-likelihood functional, $\alpha\ge 0$ is a smoothing parameter, and $\calR$ is a convex lower semi-continuous functional which regularizes the estimator. Both $\alpha$ and $\calR$ are chosen by the statistician. The choice of $\calR$ can consider a priori information about $f_\bbeta$ like smoothness or sparsity. Details of the method and examples for $\calR$ will be presented in Section \ref{sec:method}. Estimating $f_\bbeta$ from observations of $Y_i$ and $\bX_i$ is an ill-posed inverse problem. The degree of ill-posedness depends on the tail behavior of the design density $f_\bX$. 

In this paper we give conditions for method \eqref{eqn:method1} to achieve the optimal rates for design densities with different tail behavior or with bounded support. 
We evaluate the the small sample performance in simulations and have a data application in consumer demand.

Nonparametric estimation in model \eqref{eqn:rc_model} with one regressor, i.e. $d=1$ was first developed by \cite{Beran1992, Beran1996, Feuerverger2000}. The optimal rate for design densities with Cauchy-type tails is derived in \cite{hoderlein2010} for a kernel method. For the special case of just one regressor, i.e. $d=1$, optimal rates for tails with polynomial decay are presented in \cite{HM:20}. Faster decay of the tails and even bounded support for one regressor was considered in \cite{Hohmann2016}. Testing in model \eqref{eqn:rc_model} has been considered in \cite{Breunig2018, Dunker2019}.

Regularized maximum likelihood methods have been considered for other problems than model \eqref{eqn:rc_model}, e.g. in \cite{WH:12,HW:13, hohage2016inverse, Dunker2014}. We also refer to the literature on general convex regularization functionals, see \cite{Scherzer_etal:09, SKHK:12, benning_burger_2018} for an overview.

Random coefficients modmodels are frequently used to evaluate panel data, cf. \cite{Hsiao14} or \cite{Hasio04}, Chapter 6, for an overview. Modeling and estimating consumer demand in industrial organization and marketing often makes use of random coefficients \cite{BLP:95, Petrin:02, Nevo:01, Dube12}. In all these works, parametric assumptions on $f_{\bbeta}$ are imposed. els are applied in econometrics, epidemiology, and quantum mechanics. Applications in epidemiology are considered by \cite{Greenland00, Greenland06}. In economics, random coefficients Recently, nonparametric approaches for random coefficients became popular in microeconometrics \cite{hoderlein2010, Masten18, HHM:17}, frequently combined with binary choice \cite{Ichimura98, Berry07, Gautier12, Gautier13, Masten16, DHK:13, Fox16,DHK:17, DHKS:18}, among others. The random coefficients model also appears in quantum homodyne tomography, see \cite{Feuerverger2000, butucea2007}.

The paper is organized as follows. Section \ref{sec:method} gives a detailed account of the method. A convergence rate result is presented in Section \ref{sec:convergence}. We explain our implementation of the method in Section \ref{sec:implementation} and show simulation results in Section \ref{sec:simulation}. Section \ref{sec:application} presents an application to real data. The paper closes with a conclusion in Section \ref{sec:conclusion}.

\section{Method} \label{sec:method}

Throughout this paper we assume that $\bbeta$ has a Lebesgue density $f_\bbeta$. If in addition the conditional density $f_{Y|\bX}$ exists, the two densities are connected by the integral equation

\begin{align}
	f_{Y|\bX}(y|\bX=\bx)=\int_{\R^{d}}\mathbbm{1}\big\{\bb^\top \bx=y\big\}f_\bbeta(\bb)d\mu_d(\bb) = \int_{\bb^\top \bx=y}f_\bbeta(\bb)d\mu_d(\bb)
\label{eqn:line_int_f_beta}
\end{align}

where $\mathbbm{1}$ denotes the index function and $\mu_d$ the induced Lebesgue measure on the $d$-dimensional hyper-plane given by $\bb^\top \bX=Y$. It was pointed out in \cite{Beran1996} that \eqref{eqn:line_int_f_beta} becomes the well known Radon transformation $f_{S|\bTheta}=\calR(f_\bbeta)$ when $\bX$ and $Y$ are normalized to $\bTheta = \bX/\|\bX\|$ and $S = Y/\|\bX\|$. The injectivity of the Randon transformation implies that $f_\bbeta$ is identified when the support of $\bTheta$ contains an open neighborhood, see \cite{Masten18} for minimal conditions.

It is well known that an integral equation of the first kind like \eqref{eqn:line_int_f_beta} and the Radon transformation are ill-posed inverse problems, i.e. $f_\bbeta$ does not depend continuously on $f_{Y|\bX}$ or $f_{S|\bTheta}$. Nonparametric estimation of $f_\bbeta$ requires regularization. It is possible to derive a nonparametric estimate for $f_\bbeta$ by inverting the Radon transformation on a nonparametric estimator $\hat f_{S|\bTheta}$. We want to avoid this two step approach and suggest an estimator that allows the statistician to use a priori knowledge of $f_\bbeta$ if available.

%
%
%
 
Our method is based on maximum likelihood estimation. The likelihood functional for $f_\bbeta$ given the observations is
\begin{equation}
	\mathcal{L}(f_\bbeta|Y,\bX)=\prod_{i=1}^{n}\int_{\mathbbm{R}}\mathbbm{1}\big\{\bbeta_i^\top \bX_i=Y_i\big\}f_\bbeta(\bb)d\mu(\bb),
	\label{eqn:likelihood_func}
\end{equation}
which leads to the expression for the average log-likelihood
\begin{equation}
	\bar{\ell}(f_\bbeta|Y,\bX)=\frac{1}{n}\sum_{i=1}^{n}\log\left[\int_{\mathbbm{R}}\mathbbm{1}\big\{\bbeta_i^\top \bX_i=Y_i\big\}f_\bbeta(b)d\mu(\bb)\right].
	\label{eqn:avg_lhood}
\end{equation}
Direct maximization of $\bar\ell(f_\bbeta|Y,\bX)$ over all densities it not feasible due to ill-posedness and will lead to overfitting. We stabilize the problem by adding a penalty term $\alpha \calR(f_\bbeta)$ and state the estimator as a minimization problem with negative log-likelihood
\begin{align}\label{eqn:method}
\hat {f_\bbeta}_\alpha = \argmin_{f\ge 0,\, \|f\|_{L^1}=1} -\bar\ell(f|Y,\bX)+\alpha \calR(f).
\end{align}
Here $\alpha \ge 0$ is a regularization parameter that controls a bias variance trade-off. 
Furthermore, $\calR$ is a convex lower semi-continuous functional. Typical choices for $\calR$ are

\begin{itemize}
\item Squared $L^2$ norm: $\calR(f) = \|f\|_{L^2}^2 = \|f\|_{2}^2$.

\item Sobolev Norm for $H^{r}$: $ \calR(f) =  \|f_\beta\|_{2}^{2}+\|f^{\prime}_\beta\|_{2}^{2}+\ldots+\|f^{(r)}_\beta\|_{2}^{2}$.

\item $L^p$ norm: $\calR(f) = \|f\|_{L^p}^p = \|f\|_{p}^p $ with $0\le p \le \infty$.

\item Sobolev Norm for $W^{r}_p$: $ \calR(f) =  \|f_\beta\|_{p}^{p}+\|f^{\prime}_\beta\|_{p}^{p}+\ldots+\|f^{(r)}_\beta\|_{p}^{p}$.

\item Entropy: $\calR(f) = \int{f(\bb)\ln{f(\bb)}d\bb}$.
\end{itemize}
While an $L^p$ penalty implies that the estimator $\hat f_\bbeta$ is in $L^p$, a Sobolev norm enforces more smoothness on the estimator. Note that an entropy penalty is natural for estimating a density. It does not imply that $f_\bbeta$ belongs to a particular $L^p$ space other than $L^1$ and only enforces a finite entropy. If a priori knowledge about $f_\bbeta$ like smoothness is available, $\calR$ can be chosen accordingly.

In addition to the regularization functional the regularization parameter $\alpha$ needs to be chosen. As usual in nonparametric estimation the optimal $\alpha$ depends on the unknown $f_\bbeta$. We propose adaptive estimation by Lepskii's principle which is popular for inverse problems, e.g. \cite{Tsybakov:00}, \cite{BauHoh:05}, \cite{mathe_2006}, \cite{hohage2016inverse}, \cite{Werner:18}, and \cite{HM:20}. It is computationally cheaper than many other parameter selection rules like cross-validation, which is a big advantage since our problem is already computationally demanding.

The principle works as follows: we compute the estimator $\hat f_{\bbeta_{\alpha_1}},\ldots,\hat f_{\bbeta_{\alpha_m}}$ for several parameters $\alpha_1 < \alpha_1 <\ldots < \alpha_m$ together with a bound on the standard deviation of each estimator which is a non-increasing function $\Psi:\{i,\ldots,m\} \rightarrow \mathbbm{R}$ which is specified in Theorem \ref{thm:lepskii}. The Lepskii regularization parameter  $\alpha_{\bar{j}}$ is selected by
\begin{align*}
\bar{j}:=\max\{j\leq m, \|{f_\bbeta}_{\alpha_i},{f_\bbeta}_{\alpha_j}\|\leq 2\Psi(i), \text{ for all } i \leq j\}.
\end{align*}
For further details on this parameter choice we refer to \cite{Lepskii:91},\cite{mathe_2006}, and \cite{hohage2016inverse}.

Our method can be interpreted as an elastic net with $L^1$ penalty combined with the penalty term given by $\calR$. It was poined out in \cite{heiss2019nonparametric} that the constraint $\|f_\bbeta\|_{L^1}=1$ together with a finite difference discretization of $f_\bbeta$ is equivalent to an $\ell^1$ penalty on the discretized values of $f_\bbeta$. Without the additional penalty term $\alpha\calR(f_\bbeta)$ this would lead to an unwanted shrinkage in the estimate which would set $\hat f_\bbeta$ to $0$ at most grid points and produce relatively large values for $\hat f_\bbeta$ at only a few grid points. To reduce this LASSO-type effect \cite{heiss2019nonparametric} introduce an additional quadratic constraint which transforms the method into an elastic net. This is analogous to our regularization term $\alpha\calR(f_\bbeta)$ although our setup is more flexible since it allows for a wide range of regularization functionals and not only for quadratic $\calR$.


\section{Convergence rates}\label{sec:convergence}


In this section we establish some theoretical properties of the estimator, such as the existence, and uniqueness of the solution to \eqref{eqn:method}, as well as its convergence properties. For technical reasons, we will state the theory for the normalized variables $\bTheta = \bX/\|\bX\|$ and $S = Y/\|\bX\|$. Note that the estimator $\hat f_\bbeta$ in \eqref{eqn:method} does not change when we replace $\bX$ by $\bTheta$ and $Y$ by $S$.


\subsection{Assumptions}

We define the linear compact integral operator 
\[
(Tf)(\btheta,s) = \int_{\R^{d}}\mathbbm{1}\big\{\bb^\top \btheta=s\big\}f(\bb)d\mu_d(\bb).
\] 
Hence, we are solving the inverse ill-posed problem
\begin{equation}
f_{S|\bTheta} = T f_\beta
\label{eq:inv_prob}.
\end{equation}
We start by discussing the assumptions.

\begin{assump}\label{as:model}
We have i.i.d. samples $S_i,\bTheta_i$, $i=1,\ldots,n$ generated by model \eqref{eqn:rc_model} with independent $\bbeta_i$ and $\bTheta_i$. Let $B \subset L^1(\R^{d+1})$ be convex with $f_\bbeta \in B$ and $T:B \rightarrow H^s(D)$ where $H^s$ is an $L^2$ based Sobolev space with $s>\frac{d+1}{2}$, and $D\subset\R^{d+1}$ is some bounded Lipschitz domain.
\end{assump}

\begin{assump}\label{as:R}
$\mathcal{R}: B \rightarrow \R$ is convex, and lower semi-continuous.
\end{assump}


Assumption \ref{as:model} is a technical assumption that we need to derive convergence rates. It is a mild smoothness assumption on the $Tf_\bbeta$ and needs $Tf_\bbeta$ to have bounded support. This implies that also $f_\bbeta$ has bounded support. Assumption \ref{as:R} holds for many regularization terms, e.g. for the examples listed in Section \ref{sec:method}. For more examples see \cite{hohage2016inverse}.

\bigskip

\begin{thm}\label{thm:exist} 
If the assumptions \ref{as:model} and \ref{as:R} hold, then \eqref{eqn:method} has a unique minimizer $\hat{f_\beta}_{\alpha}$ for all $\alpha >0$.
\end{thm}

As a smoothness assumption for the convergence rate results we use a variational source condition which has been used extensively in the inverse problems literature, e.g. \cite{flemming2010GeneralResid}, \cite{WH:12}, \cite{Dunker2014}, \cite{hohage2016inverse}. It uses Bregman distance between some density $f$ and $f_\bbeta$ with respect to $\calR$ and some $f^*_\bbeta \in \partial\calR(f_\bbeta)$ where $\partial\calR(f_\bbeta)$ is the subdifferential of the convex functional $\calR$ at $f_\bbeta$. The Bregman distance with respect to $\calR$ 
and $f^*_\bbeta$ is defined as
\[
D_\calR^{f^*_\bbeta}(f,\, f_\bbeta) := \calR(f) - \calR(f_\bbeta) - \langle f^*_\bbeta, f - f_\bbeta \rangle.
\]
For quadratic penalty in Hilbert spaces $\calR=\|\cdot\|^2$ we have 
$D_\calR^{f^*_\bbeta}(f,\, f_\bbeta)= \|f-f_\bbeta\|^2$. 
In general, $D_{\calR}^{f^*_\bbeta}$ is nonnegative with 
$D_{\calR}^{f^*_\bbeta}(f_\bbeta,f_\bbeta)=0$, but it is not necessarily 
symmetric or satisfies a triangle inequality. We will also use the Kullback-Leibler divergence $\KL(f,g)= \int f \ln\left(f/g\right) \, dx$.

\begin{assump}[Variational Source Condition] There exists $f_\beta^{*} \in \partial{R(f_\beta)}$, $c > 0$ and a concave, monotonically increasing function $\lambda:[0,\infty)\rightarrow [0,\infty)$ with $\lambda(0)=0$ such that
	\begin{equation}
	cD_{\calR}^{f_\beta^{*}}(f,f_\bbeta) \leq \calR(f) - \calR(f_\bbeta) + \lambda(\KL(f_{S|\bTheta};Tf))
	\end{equation}
for all densities $f$.
\label{asmp:var_ineq}
\end{assump}

The function $\lambda$ is called index function. 
The rate with which $\lambda$ increases compares the smoothing properties of $T$ with the smoothness of $f_\bbeta$. The slower $\lambda$ increases the faster the convergence of the method as we will see in Theorem \ref{thm:aprior_conv}.

Assumption \ref{asmp:var_ineq} is an abstract assumption which does not refer to a classical smoothness class like Sobolev or H\"older spaces. The popularity of this assumption in the literature comes from the fact that it is not only sufficient but also necessary for a certain convergence rate, \cite{Flemming:10,Flemming:12,flemming:12b}. In other words, the class of functions which fulfil Assumption \ref{asmp:var_ineq} with a specific $\lambda$ is the largest class of functions for which the convergence rate in Theorem \ref{thm:aprior_conv} is achieved. For several operators and regularisation functionals variatioinal source conditions are implied by classical smoothness assumptions, see \cite{WeHo:17, HoMi:19, WeSpHo:20}. Assumption \ref{asmp:var_ineq} hold for $f_\bbeta$ in an $L^2$ based Sobolev $H^r$ space in the following special case.
\begin{thm}
\label{thm:sc}
Let $f_\bbeta \in H^r(\R^n)$ with $0<r\le d-1$. Assume that $f_{S|\bTheta}$ is bounded and that $f_\bTheta$ is bounded and bounded away from zero. On the function class $\mathcal{T} = \{ f \in L^1(\R^d) | Tf \in L^2,  \|Tf\|_\infty\le c_f \}$ with some $c_f>0$ Assumption \ref{as:R} holds true for $\calR = \|\cdot\|^2_{L^2}$ with $\lambda(x) = x^\frac{2r}{2r+d-1}$. More generally, for $\calR = \|\cdot\|^2_{H^a}$ with $0 < r-a \le d-1$ Assumption \ref{as:R} holds true on $\mathcal{T}$ with $\lambda(x) = x^\frac{2(r-a)}{2(r-a)+d-1}$.
\end{thm}

\subsection{Convergence Rate Results}

The following theorem establishes the convergences rates with a priori parameter choice.
\begin{thm}
\label{thm:aprior_conv} If assumptions \ref{as:model} and \ref{as:R} hold, and if $\alpha$ is chosen as
	\begin{equation}
		\alpha^{-1} \in -\partial(-\lambda)\left(\frac{1}{\sqrt{n}}\right),
		\label{eq:param_rule}
	\end{equation}
we obtain the following convergence rate
	\begin{equation*}
		\mathbbm{E}\big[D_{R}^{f_\bbeta^{*}}(\hat{f_\bbeta}_{\alpha},f_\bbeta)\big]= \mathcal{O}\left(\lambda\left(\frac{1}{\sqrt{n}}\right)\right), \;\; n \mapsto \infty
	\end{equation*}
\end{thm}

The last theorem gives convergence rates in terms of the Bregman distance with respect to $\calR$. This is an unusual convergence measure but it comes with easy interpretations in many cases. If for example $\calR$ is maximum entropy regularization, the Bregman distance dominates the $L^1$ norm $ D_{R}^{f_\bbeta^{*}}(f,f_\bbeta) > \|f-f_\bbeta\|_{L^1}$. Hence, Theorem \ref{thm:aprior_conv} implies convergence rates for the $L^1$ norm. If $\calR$ is a quadratic penalty like $\calR = \|\cdot\|_{L^2}^2$ or $\calR = \|\cdot\|_{H^a}^2$, the Bregman distance coincides with the squared norm $ D_{R}^{f_\bbeta^{*}}(f,f_\bbeta) > \|f-f_\bbeta\|_{H^a}^2$. This leads to the following lemma.

\begin{thm}\label{thm:apriori}
Let the assumptions of Theorem \ref{thm:sc} hold and use the a priori parameter choice $\alpha = \frac{2(r-a)+d-1}{2(r-a)	}n^\frac{1-d}{4(r-a)+2(d-1)}$. Then
\[
\mathbbm{E} \big[ \|\hat{f_\bbeta}_{\alpha}-f_\bbeta\|_{H^a}^2\big] = \mathcal{O}\left( n^{-\frac{r-a}{2(r-a)+d-1}} \right).
\]
\end{thm}

Note that the theorem proves convergence of the mean integrated squared error for $a=0$. Hence, $MISE(\hat{f_\beta}_{\alpha}) = \mathcal{O}\big( n^{-\frac{r}{2r+d-1}} \big)$ with parameter choice $\alpha = c n^\frac{1-d}{4(r-a)+2(d-1)}$ with some constant $c>0$.

The choice of $\alpha$ in Theorems \ref{thm:aprior_conv} and \ref{thm:apriori} requires a priori knowledge about the smoothness of the unknown $f_\bbeta$. This is usually not available in practice. A popular and computationally inexpensve data-driven parameter choice for $\alpha$ is Lepskii's principle as in \cite{Lepskii:92b}, \cite{mathe_2006}, and \cite{hohage2016inverse}. For sake of simplicity we state the result only for quadratic penalty in Hilbert spaces $\calR=\|\cdot\|$. 

The principle works as follows: we compute $\hat{f_\bbeta}_{\alpha_{1}},\ldots,\hat{f_\bbeta}_{\alpha_{m}}$ for $\alpha_1 = c_L \frac{\ln(n)}{\sqrt{n}}$ and $\alpha_{i+1} = r\alpha_i$ with some constants $c_L>0,r>1$.

\[
\bar{j}:=\min\{j\leq m, \|\hat{f_\bbeta}_{\alpha_i}-\hat{f_\bbeta}_{\alpha_j}\| \leq 8 r^\frac{1-i}{2}), \text{ for all } i < j\}.
\]

\begin{thm}\label{thm:lepskii}
Let $\calR = \|\cdot\|_{H^a}^2$ be a Sobolev regularization functional. If Assumption \ref{asmp:var_ineq} holds with $c\ge\frac{1}{2}$, then with the Lepskii parameter choice $\alpha_{\bar j}$
\begin{equation}
	\mathbbm{E}[\|\hat{f_\bbeta}_{\alpha_{\bar j}} - f_\bbeta\|_{H^a}^2]= \mathcal{O}\left(\lambda\left(\frac{\ln(n)}{\sqrt{n}}\right)\right).
	\label{eq: lep_conv}
\end{equation}
If in addition the assumptions of Theorem \ref{thm:sc} hold, then
\[
\mathbbm{E} \big[ \|\hat{f_\bbeta}_{\alpha}-f_\bbeta\|_{H^a}^2\big] = \mathcal{O}\left(\ln(n)^{\frac{2(r-a)}{2(r-a)+d-1}}~ n^{-\frac{r-a}{2(r-a)+d-1}} \right).
\]
\end{thm}

%
%

As usual we lose a log-factor for adaptive estimation. Other \textit{a posteriori} methods might perform better than Lepskii's principle e.g. cross-validation. However, these methods are computationally expensive and not feasible in our setup. Hence, we favor Lepskii's principle for its simplicity and computational efficiency.

\section{Implementation}\label{sec:implementation}
\subsection{Discretization}

We discretized $f_\beta$ over a fixed regular grid with $m$ grid points. Applying this discretization allowed for the evaluation of the integral expression $\int_{\mathbbm{R}}\mathbbm{1}\big\{\bbeta_i^\top\bX_i=Y_i\big\}f_\bbeta(\bb)d\bb$ using a finite volume approach. This integral expression is simply the integral of the joint density $f_\beta$ over the $d$-dimensional hyperplanes parametrized by the sample points $Y_i=\bb_i^\top\bX_i$. In this implementation, the model was specified to have one regressor in which case the hyperplanes are simply lines.

The amount of information that can be known about the shape of $f_\bbeta$ relies on the amount of coverage that the lines parametrized by the sample observations have over the estimation grid. Figure 1 illustrates this concept for the case of the random coefficients model with one regressor and random slope, with random design.

\begin{equation}
	y_i={\beta_0}_i+{\beta_1}_{i}{x_1}_i
	\label{eqn:rc_model_2d}
\end{equation}

It is clear from equation \eqref{eqn:rc_model_2d} that the coverage of the lines over the estimation grid depends highly on the design density as it specifies the slopes of the different lines parametrized by the sample. In section \ref{subsec: Bounded} we will discuss how our method is more robust with respect to poor angle coverage which can be brought about by either undesired tail behavior in the distribution of the regressor, or the said regressor having too narrow of a support.

\begin{figure}[ht!]
	\centering
	\includegraphics[width=90mm]{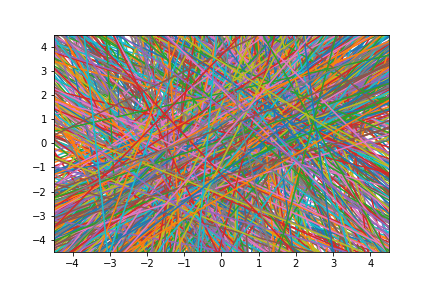}
	\caption{Example line coverage of the estimation grid. Lines are parametrized as $Y_i=\bb^\top \bX_i$, or by $\beta_0=y_i-\beta_1 x_{1_i}$. \label{fig:full_angle}}
\end{figure}

\subsection{Optimization of the Likelihood Functional}

The discretized likelihood functional is a $m$-dimensional function, ${\bar{\ell}}^{*}: \mathbbm{R^{m}}\mapsto\mathbbm{R}$. We denote the discretized form of the regularization functionals as, ${\calR}^{*}({f^{*}}): \mathbbm{R^{m}}\mapsto\mathbbm{R}$, where $f^{*} = \{f_{j}\}_{j=1}^{m}$

The nonlinear program for the minimization of the conditional average log-likelihood functional was setup as follows:\\

\begin{equation}
\hat{f_\bbeta} = \argmin_{f^{*}}-{\bar{\ell}^{*}}(f^{*}|X,Y)+\alpha \calR^{*}(f^{*})
\end{equation}

{%
	\centering
	$\begin{array}{rcl}
	\text{subject to}: \sum_{j=1}^{m}f_{j}\Delta b=1,
	\\f_{1},f_{2},\ldots,f_{m} \geq 0.\\
	\\
	\end{array}$
	\\
}

The convexity of the problem allows us to use a trust-region constrained algorithm based on an interior point algorithm for constrained large-scale nonlinear programs developed in \cite{byrd1999interior}. This particular algorithm excels in solving high-dimensional minimization problems, which is the case for the implementation as $f_\beta$ can scale up to the order of thousands in terms of dimensionality. We used the Python implementation of the algorithm in \cite{byrd1999interior}.

\section{Monte Carlo Simulations}\label{sec:simulation}

We checked the small sample performance of the method in Monte Carlo simulations with two different test examples, one with a normally distributed regressor, and one with a regressor which has bounded support. We compared the results to a kernel-based estimator akin to the one proposed in \cite{hoderlein2010}. 
%
We found that the kernel based estimator in our implementation required us to
 cut-off negative point-estimates and to normalize the final reconstruction, $\hat{f_\beta}$, otherwise it produced very poor results.

To test the performance of both methods we conducted a simulation study by generating data for the Random Coefficients model with random slope, and one regressor as specified by equation \ref{eqn:rc_model_2d}. We then evaluated its ability to reconstruct the underlying distribution of the random coefficients by comparing it to the known true distribution through measuring the empirical mean integrated squared error (MISE) across one-hundred Monte Carlo simulation runs which we repeat for the different sample sizes. We also considered in our study the variance of the integrated squared error (ISE), as we believe stability of the reconstructions is as important to consider as their accuracy.

\subsection{Simulation Results for Unbounded Support}\label{subsec: Unbounded}

We compared the ability of both estimators to reconstruct $f_\bbeta$ in the case where the design density has unbounded support. Figure \ref{fig:full_angle} in Section \ref{sec:implementation} illustrates the case for the full angle coverage.

The data used for these set of simulations was constructed by drawing $X_1$ from a normal distribution, $X_1 \sim \mathcal{N}(0,1)$. The coefficients $\beta_0$ and $\beta_1$ were sampled from a bimodal, bivariate normal mixture distribution specified as,
\begin{equation}
0.5\mathcal{N}([-1.5,-1.5],\Sigma)+0.5\mathcal{N}([1.5,1.5],\Sigma).
\label{eq:beta_dist}
\end{equation}

\noindent with the simple covariance structure where $\Sigma$ is the identity matrix. Lastly, the variable $Y$ was created using the model equation $y_i={\beta_0}_i+{\beta_1}_i {x_1}_i$, hence producing the simulated dataset. Figure \ref{fig:ftheo} shows the shape of the density $f_\bbeta$.


\begin{figure}[ht!]
	\centering
		\includegraphics[width=90mm]{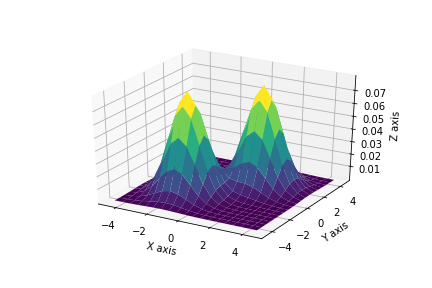}
	\caption{Theoretical $f_\beta$. \label{fig:ftheo}}
\end{figure}

\subsubsection{Simulation Results}

The simulation study was done with sample sizes $n =$ \{500, 1000, 1500, 3000, 10000\}. We used an equidistant grid on the domain $[-4.5,4.5] \times [-4.5,4.5]$ with $m = 361$ grid points for $\hat{f_\bbeta}$.

As stated in Section \ref{sec:convergence} we select the parameter $\alpha$ for the RMLE through Lepskii's principle. 
%
However, this parameter choice method did not work well for the bandwidth of the kernel method. So instead we present the results from choosing the reconstruction that has the smallest ISE for each simulation assuming that we have a priori knowledge of the true distribution $f_\bbeta$. This means that when comparing the simulation results between the two methods, the kernel-method has a significant advantage over RMLE. 

Table \ref{tab:unbounded_results} contains the  statistics 100 Monte Carlo simulation results.

\begin{table}[h]
	\caption{Monte Carlo Simulation Results for Unbounded Support}
	\centering 
	\begin{tabular}{c rrrr} 
		\hline\hline \\ [-1.5ex]
		\multicolumn{1}{}{} & \multicolumn{2}{c}{RMLE} & \multicolumn{2}{c}{kernel Method}\\ \cline{2-3} \cline{4-5}\\ [-2ex]
		n & MISE & Variance (ISE) & MISE & Variance (ISE)\\ [0.5ex]
		\hline 
		500 & 4.043e-05 & 2.171e-11 & 9.516e-06 & 3.361e-12 \\ 
		1,000 & 2.201e-05 & 1.566e-10 &6.798e-06 & 1.613e-12\\
		1,500 & 1.928e-05 & 1.751e-10 &5.720e-06 & 8.287e-12\\
		3,000 & 1.399e-05 & 4.166e-11 & 4.429e-06& 3.675e-13\\
		10,000 & 9.610e-06 & 1.926e-12 & 3.536e-06 & 8.038e-14\\
	\hline 
	\end{tabular}
\label{tab:unbounded_results}
\end{table}

In the case of the RMLE method the MISE, as well as the variance of the ISE, displayed a decreasing trend as sample size increased. The same is true for the kernel method, as a similar decreasing trend is observed for both the MISE and the variance.

Based on the simulation results, the kernel method seems to outperform the RMLE method for the unbounded case with the caveat that these results for the kernel-method were obtained using a superior \emph{a priori} parameter choice, while the RMLE method used a data-driven parameter choice for $\alpha$.

Figures \ref{fig:mle_unbounded} and \ref{fig:kernel_unbounded} are reconstructions using the RMLE method and the kernel method respectively for the unbounded case. The ISE for these reconstructions are 9.136e-06 and 3.149e-06 respectively. It is clear that both methods work well in the case where the support of the distribution of $X_1$ is unbounded. The contour plots in Figures \ref{fig:mle_unbounded_cont} and \ref{fig:kernel_unbounded_cont} of typical reconstructions suggest the same result.

\begin{figure}[h]
	\centering
	\begin{minipage}{0.4\textwidth}
		\centering
		\includegraphics[width=1.1\textwidth]{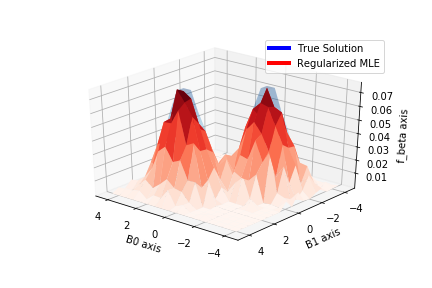} 
		\caption{$\hat{f_\beta}$ estimate for $n=10,000$ using RMLE}\
		\label{fig:mle_unbounded}
	\end{minipage}\hfill
	\begin{minipage}{0.4\textwidth}
		\centering
		\includegraphics[width=1.1\textwidth]{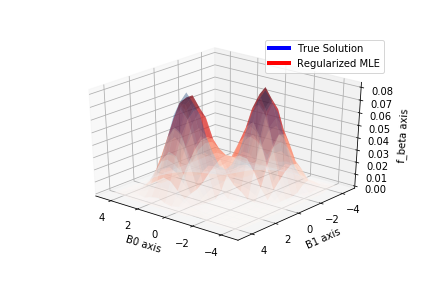} 
		\caption{$\hat{f_\beta}$ estimate for $n=10,000$ using the kernel method}
		\label{fig:kernel_unbounded}
	\end{minipage}
\end{figure}

\begin{figure}[h]
	\centering
	\begin{minipage}{0.45\textwidth}
		\centering
		\includegraphics[width=1.1\textwidth]{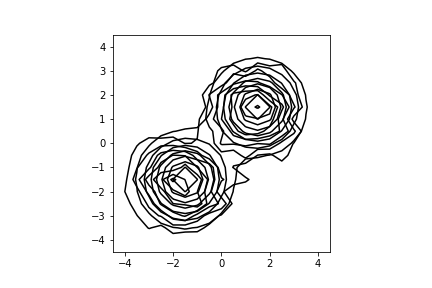} 
		\caption{$\hat{f_\beta}$ contour using RMLE}\
		\label{fig:mle_unbounded_cont}
	\end{minipage}\hfill
	\begin{minipage}{0.45\textwidth}
		\centering
		\includegraphics[width=1.1\textwidth]{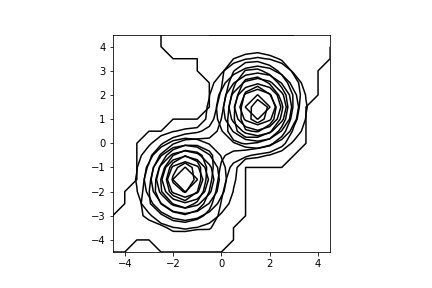} 
		\caption{$\hat{f_\beta}$ contour using the kernel method}
		\label{fig:kernel_unbounded_cont}
	\end{minipage}
\end{figure}
\subsection{Simulation Results for Bounded Support}\label{subsec: Bounded}

To illustrate the advantage of using RMLE over its kernel-based counterpart we conducted simulations for the case where the design density has bounded support.

To generate the data, the regressor $X_1$ was sampled from a uniform distribution on $\small[-0.8,0.8\small]$. The coefficients $\beta_0$ and $\beta_1$ were sampled from the same bimodal, bivariate normal mixture distribution specified in equation \eqref{eq:beta_dist}. Figure \ref{fig:lim_angle} illustrates the limited angle problem caused by the boundedness of $X_1$.

\begin{figure}[h]
	\centering
		\includegraphics[width=70mm]{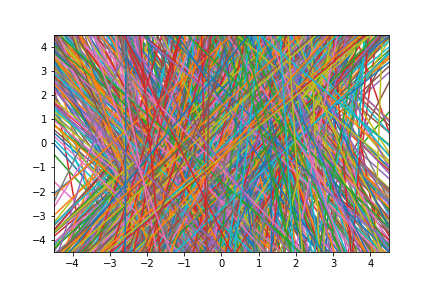}
	\caption{Limited angle coverage of the estimation grid. Lines are parametrized as $Y_i=\bbeta^\top \bX_i$, or by $\beta_0=y_i-\beta_1 x_{1_i}$. \label{fig:lim_angle}}
\end{figure}

\subsubsection{Simulation Results}
The simulation study was done with sample sizes $n =$ \{500, 1000, 1500, 3000, 10000\} and on the same grid for $\hat{f_\bbeta}$ as above. Table \ref{tab:bounded_results} contains the statistics for both estimators based on 100 Monte Carlo simulations for each sample sizes.

\begin{table}[h]
	\caption{Monte Carlo Simulation Results for Bounded Support}
	\centering 
	\begin{tabular}{c rrrr} 
		\hline\hline \\ [-1.5ex]
		\multicolumn{1}{}{} & \multicolumn{2}{c}{RMLE} & \multicolumn{2}{c}{Kernel Method}\\ \cline{2-3} \cline{4-5}\\ [-2ex]
		n & MISE & Variance (ISE) & MISE & Variance (ISE)\\ [0.5ex]
		\hline 
		500 & 4.333e-05 & 1.926e-12 & 3.320e-05 & 4.9868e-12 \\ 
		1,000 & 3.367e-05 & 8.761e-11 &3.173e-05 & 2.582e-12\\
		1,500 & 3.007e-05 & 8.292e-11 &3.089e-05 & 2.121e-12\\
		3,000 & 2.874e-05 & 5.730e-11 & 3.009e-05 & 1.229e-12\\
		10,000 & 1.689e-05 & 2.172e-11 & 2.922e-05 & 5.366e-13\\
		\hline 
	\end{tabular}
	\label{tab:bounded_results}
\end{table}

The RMLE method shows a decreasing trend in the MISE and the variance of the ISE as the sample size increases. It is noteworthy that the MISEs produced by the RMLE method in the bounded case are not much different to the ones the method produced in the unbounded case.

However, the same cannot be said for the kernel-method. Although it has a smaller variance compared to the RMLE method, the MISEs only decrease marginally as the sample size increases. When compared to the results for the unbounded case, there is a significant difference in the size of the errors produced by the method.

In the case where the design density has bounded support, it is apparent that the Regularized MLE estimator performs considerably better than its kernel-based counterpart in terms of accuracy measured by MISE even with the RMLE method using a data-driven parameter choice, and the kernel-method using an \emph{a priori} choice for the bandwidth parameter.

Figures \ref{fig:mle_bounded} and \ref{fig:kernel_bounded} show typical density reconstructions using RMLE and the kernel method with $n=10,000$. The ISEs for the reconstructions were 1.309e-05 and 2.872e-05 respectively. While the kernel method produces a much smoother reconstruction, it suffers from the limited angle problem as seen in the ridging effect that manifests along the pre-dominant directions established by the lines parametrized by the sample. Figures \ref{fig:mle_bounded_cont} and \ref{fig:kernel_bounded_cont} show the contours of the density reconstructions, $\hat{f}_\bbeta$.

\begin{figure}[h]
	\centering
	\begin{minipage}{0.4\textwidth}
		\centering
		\includegraphics[width=1.1\textwidth]{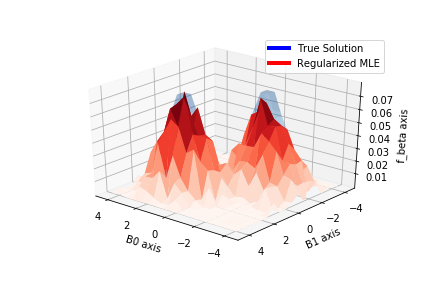} 
		\caption{$\hat{f_\beta}$ estimate for $n=10,000$ using RMLE}\
		\label{fig:mle_bounded}
	\end{minipage}\hfill
	\begin{minipage}{0.4\textwidth}
		\centering
		\includegraphics[width=1.1\textwidth]{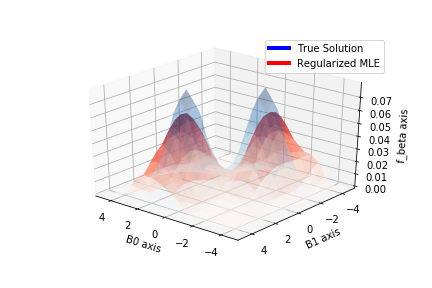} 
		\caption{$\hat{f_\beta}$ estimate for $n=10,000$ using the kernel method}
		\label{fig:kernel_bounded}
	\end{minipage}
\end{figure}

As stated previously, the parameter $\alpha$ can either be chosen by the statistician or by a data-driven parameter choice rule. If more smoothness in the reconstruction is desired it can be achieved by choosing a larger $\alpha$, or by choosing a regularization term that imposes greater smoothness constraints. 

\begin{figure}[h]
	\centering
	\begin{minipage}{0.45\textwidth}
		\centering
		\includegraphics[width=1.1\textwidth]{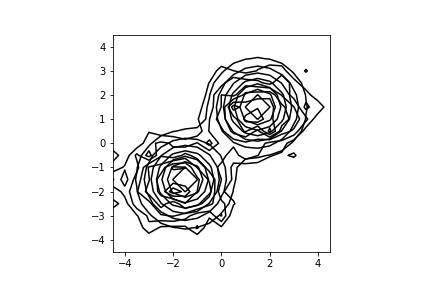} 
		\caption{$\hat{f_\beta}$ contour using RMLE}\
		\label{fig:mle_bounded_cont}
	\end{minipage}\hfill
	\begin{minipage}{0.45\textwidth}
		\centering
		\includegraphics[width=1.1\textwidth]{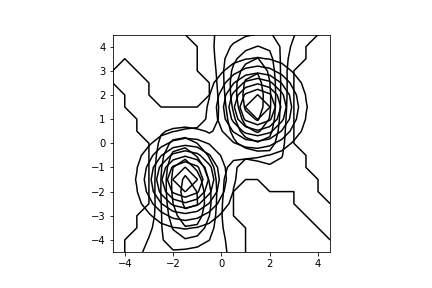} 
		\caption{$\hat{f_\beta}$ contour using the kernel method}
		\label{fig:kernel_bounded_cont}
	\end{minipage}
\end{figure}
 
We can conclude that the RMLE method works well in both cases, whereas the kernel method suffers from the problem posed by limited-angle coverage. This robustness with respect to the nature of the data allows us to work with a wider range of data without any complications; furthermore, data with limited support are quite typical in applications for these methods.

\section{Application}\label{sec:application}

For a real data application of the estimator we use the data set from \cite{Dunker2019}, which are data from the British Family Expenditure Survey from 1997 to 2001 with a sample size of about 33000. 
The variables we are concerned with being the budget share of food, the log of total expenditure, and the log of food prices. The same transformation on the data done in \cite{Dunker2019} was applied as it leads to more coverage over the estimation grid which improves the estimator's performance. The model we apply to the data is is motivated by the almost ideal demand system \cite{Deaton:80}:
\begin{equation}
BS_i = {\beta_1}_i \ln(Total Expenditure_i) + {\beta_2}_i \ln(Food Prices_i) + \epsilon_i
\label{eq:aids_model}
\end{equation}

The Ordinary Least-Squares (OLS) results for the model return fixed coefficients as $ (\beta_1,\beta_2) = (0.085,0.017)$. Figures \ref{fig:mle_real}, and \ref{fig:mle_real_cont} show the results using the RMLE method for the same data and model. The results from our estimator suggest that the peak, of $\hat{f}_{\beta_1,\beta_2}$ occurs close to $(0,0)$. 

While the OLS results for the fixed coefficients is consistent with what our estimator suggests to be the location of the peak, what it fails to capture is the unobserved heterogeneity present in the population. And, as the result from our estimator shows there is noticeable spread in the estimates for $\beta_1$, and $\beta_2$ which suggests that a random coefficient model might be a better choice than the fixed coefficient model.

\begin{figure}[h]
	\centering
	\begin{minipage}{0.4\textwidth}
		\centering
		\includegraphics[width=1.1\textwidth]{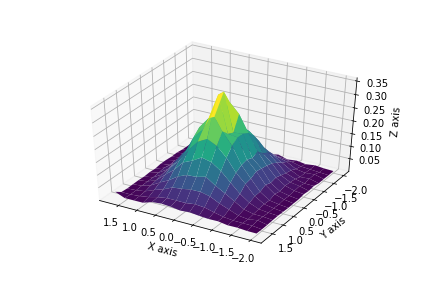} 
		\caption{$\hat{f}_{\beta_1,\beta_2}$ estimate.}\
		\label{fig:mle_real}
	\end{minipage}\hfill
	\begin{minipage}{0.4\textwidth}
		\centering
		\includegraphics[width=1.1\textwidth]{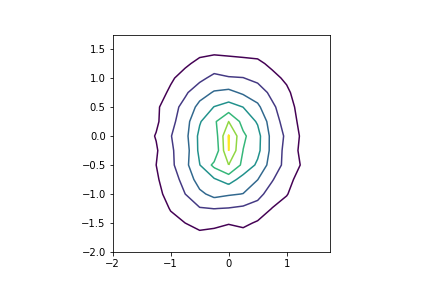} 
		\caption{Contour of $\hat{f}_{\beta_1,\beta_2}$}
		\label{fig:mle_real_cont}
	\end{minipage}
\end{figure}

\section{Conclusion}\label{sec:conclusion}
We provide an alternative method of estimating the joint distribution of $\bbeta$ through Regularized Maximum Likelihood Estimations (RMLE). To combat the instability of the estimate caused by the implicit inversion of the Radon Transform, we apply Tikhonov-type regularization methods.

Compared to its well-known counter part, the kernel method, the RMLE method is more robust with respect to tail behavior of the design density. Although more robust, the method is not completely unaffected by the behavior of the design density, as the rate of decay of the singular values of the operator,  and consequently the degree of ill-posedness of the problem, $T$ is determined by the distribution of the regressors.

Monte Carlo simulations illustrate that both methods work well in a scenario where the regressor has unbounded support, but in the case where the regressor has bounded support the RMLE method produces more accurate reconstructions. 


\bibliographystyle{apalike}
\bibliography{RC_MLE_Paper}

\begin{thebibliography}{}

\bibitem[Bauer and Hohage, 2005]{BauHoh:05}
Bauer, F. and Hohage, T. (2005).
\newblock A lepskij-type stopping rule for regularized newton methods.
\newblock {\em Inverse Problems}, 21(6):1975--1991.

\bibitem[Benning and Burger, 2018]{benning_burger_2018}
Benning, M. and Burger, M. (2018).
\newblock Modern regularization methods for inverse problems.
\newblock {\em Acta Numerica}, 27:1–111.

\bibitem[Beran et~al., 1996]{Beran1996}
Beran, R., Feuerverger, A., and Hall, P. (1996).
\newblock On nonparametric estimation of intercept and slope distributions in
  random coefficient regression.
\newblock {\em Ann. Statist.}, 24(6):2569--2592.

\bibitem[Beran and Hall, 1992]{Beran1992}
Beran, R. and Hall, P. (1992).
\newblock Estimating coefficient distributions in random coefficient
  regressions.
\newblock {\em Ann. Statist.}, 20(4):1970--1984.

\bibitem[Berry et~al., 1995]{BLP:95}
Berry, S., Levinsohn, J., and Pakes, A. (1995).
\newblock Automobile prices in market equilibrium.
\newblock {\em Econometrica}, 63(4):841--890.

\bibitem[Berry and Pakes, 2007]{Berry07}
Berry, S. and Pakes, A. (2007).
\newblock The pure characteristics demand model.
\newblock {\em Internat. Econom. Rev.}, 48(4):1193--1225.

\bibitem[Borwein and Lewis, 1991]{BL:91}
Borwein, J.~M. and Lewis, A.~S. (1991).
\newblock Convergence of best entropy estimates.
\newblock {\em SIAM J. Optim.}, 1(2):191--205.

\bibitem[Breunig and Hoderlein, 2018]{Breunig2018}
Breunig, C. and Hoderlein, S. (2018).
\newblock Specification testing in random coefficient models.
\newblock {\em Quantitative Economics}, 9(3):1371--1417.

\bibitem[Butucea et~al., 2007]{butucea2007}
Butucea, C., Guţă, M., and Artiles, L. (2007).
\newblock {Minimax and adaptive estimation of the Wigner function in quantum
  homodyne tomography with noisy data}.
\newblock {\em Ann. Statist.}, 35(2):465--494.

\bibitem[Byrd et~al., 1999]{byrd1999interior}
Byrd, R.~H., Hribar, M.~E., and Nocedal, J. (1999).
\newblock An interior point algorithm for large-scale nonlinear programming.
\newblock {\em SIAM Journal on Optimization}, 9(4):877--900.

\bibitem[Deaton and Muellbauer, 1980]{Deaton:80}
Deaton, A. and Muellbauer, J. (1980).
\newblock An almost ideal demand system.
\newblock {\em American Economic Review}, 70:312--326.

\bibitem[Dubé et~al., 2012]{Dube12}
Dubé, J.-P., Fox, J.~T., and Su, C.-L. (2012).
\newblock Improving the numerical performance of static and dynamic aggregate
  discrete choice random coefficients demand estimation.
\newblock {\em Econometrica}, 80(5):2231--2267.

\bibitem[Dunker et~al., 2019]{Dunker2019}
Dunker, F., Eckle, K., Proksch, K., and Schmidt-Hieber, J. (2019).
\newblock Tests for qualitative features in the random coefficients model.
\newblock {\em Electron. J. Statist.}, 13(2):2257--2306.

\bibitem[Dunker et~al., 2013]{DHK:13}
Dunker, F., Hoderlein, S., and Kaido, H. (2013).
\newblock Random coefficients in static games of complete information.
\newblock {\em cemmap Working Papers}, CWP12/13.

\bibitem[Dunker et~al., 2017]{DHK:17}
Dunker, F., Hoderlein, S., and Kaido, H. (2017).
\newblock Nonparametric identification of random coefficients in endogenous and
  heterogeneous aggregate demand models.
\newblock {\em cemmap Working Papers}, CWP11/17.

\bibitem[Dunker et~al., 2018]{DHKS:18}
Dunker, F., Hoderlein, S., Kaido, H., and Sherman, R. (2018).
\newblock Nonparametric identification of the distribution of random
  coefficients in binary response static games of complete information.
\newblock {\em Journal of Econometrics}, 206(1):83 -- 102.

\bibitem[Dunker and Hohage, 2014]{Dunker2014}
Dunker, F. and Hohage, T. (2014).
\newblock On parameter identification in stochastic differential equations by
  penalized maximum likelihood.
\newblock {\em Inverse Problems}, 30(9):095001.

\bibitem[Feuerverger and Vardi, 2000]{Feuerverger2000}
Feuerverger, A. and Vardi, Y. (2000).
\newblock Positron emission tomography and random coefficients regression.
\newblock {\em Ann. Inst. Statist. Math.}, 52(1):123--138.

\bibitem[Flemming, 2010]{Flemming:10}
Flemming, J. (2010).
\newblock Theory and examples of variational regularization with non-metric
  fitting functionals.
\newblock {\em J. Inverse Ill-Posed Probl.}, 18(6):677--699.

\bibitem[Flemming, 2012a]{flemming:12b}
Flemming, J. (2012a).
\newblock {\em Generalized {T}ikhonov regularization and modern convergence
  rate theory in {B}anach spaces}.
\newblock Shaker Verlag, Aachen.

\bibitem[Flemming, 2012b]{Flemming:12}
Flemming, J. (2012b).
\newblock Solution smoothness of ill-posed equations in hilbert spaces: four
  concepts and their cross connections.
\newblock {\em Applicable Analysis}, 91(5):1029--1044.

\bibitem[Flemming and Hofmann, 2010]{flemming2010GeneralResid}
Flemming, J. and Hofmann, B. (2010).
\newblock A new approach to source conditions in regularization with general
  residual term.
\newblock {\em Numerical functional analysis and optimization}, 31(3):254--284.

\bibitem[Fox and Gandhi, 2016]{Fox16}
Fox, J.~T. and Gandhi, A. (2016).
\newblock Nonparametric identification and estimation of random coefficients in
  multinomial choice models.
\newblock {\em The RAND Journal of Economics}, 47(1):118--139.

\bibitem[Gautier and Hoderlein, 2012]{Gautier12}
Gautier, E. and Hoderlein, S. (2012).
\newblock {A triangular treatment effect model with random coefficients in the
  selection equation}.
\newblock {\em cemmap Working Papers}, CWP39/12.

\bibitem[Gautier and Kitamura, 2013]{Gautier13}
Gautier, E. and Kitamura, Y. (2013).
\newblock Nonparametric estimation in random coefficients binary choice models.
\newblock {\em Econometrica}, 81(2):581--607.

\bibitem[Greenland, 2000]{Greenland00}
Greenland, S. (2000).
\newblock When should epidemiologic regressions use random coefficients?
\newblock {\em Biometrics}, 56(3):915--921.

\bibitem[Gustafson and Greenland, 2006]{Greenland06}
Gustafson, P. and Greenland, S. (2006).
\newblock The performance of random coefficient regression in accounting for
  residual confounding.
\newblock {\em Biometrics}, 62(3):760--768.

\bibitem[Heiss et~al., 2019]{heiss2019nonparametric}
Heiss, F., Hetzenecker, S., and Osterhaus, M. (2019).
\newblock Nonparametric estimation of the random coefficients model: An elastic
  net approach.
\newblock {\em arXiv preprint arXiv:1909.08434}.

\bibitem[Hertle, 1983]{hertle:83}
Hertle, A. (1983).
\newblock Continuity of the radon transform and its inverse on euclidean space.
\newblock {\em Mathematische Zeitschrift}, 184(2):165--192.

\bibitem[Hoderlein et~al., 2017]{HHM:17}
Hoderlein, S., Holzmann, H., and Meister, A. (2017).
\newblock The triangular model with random coefficients.
\newblock {\em Journal of Econometrics}, 201(1):144 -- 169.

\bibitem[Hoderlein et~al., 2010]{hoderlein2010}
Hoderlein, S., Klemel{\"a}, J., and Mammen, E. (2010).
\newblock Analyzing the random coefficient model nonparametrically.
\newblock {\em Econometric Theory}, 26(3):804--837.

\bibitem[Hofmann and Yamamoto, 2010]{HoYa:10}
Hofmann, B. and Yamamoto, M. (2010).
\newblock On the interplay of source conditions and variational inequalities
  for nonlinear ill-posed problems.
\newblock {\em Applicable Analysis}, 89(11):1705--1727.

\bibitem[Hohage and Miller, 2019]{HoMi:19}
Hohage, T. and Miller, P. (2019).
\newblock Optimal convergence rates for sparsity promoting
  wavelet-regularization in besov spaces.
\newblock {\em Inverse Problems}, 35(6):065005.

\bibitem[Hohage and Weidling, 2017]{WeHo:17}
Hohage, T. and Weidling, F. (2017).
\newblock Characterizations of variational source conditions, converse results,
  and maxisets of spectral regularization methods.
\newblock {\em SIAM Journal on Numerical Analysis}, 55(2):598--620.

\bibitem[Hohage and Werner, 2013]{HW:13}
Hohage, T. and Werner, F. (2013).
\newblock Iteratively regularized {N}ewton-type methods for general data misfit
  functionals and applications to {P}oisson data.
\newblock {\em Numer. Math.}, 123(4):745--779.

\bibitem[Hohage and Werner, 2016]{hohage2016inverse}
Hohage, T. and Werner, F. (2016).
\newblock Inverse problems with poisson data: statistical regularization
  theory, applications and algorithms.
\newblock {\em Inverse Problems}, 32(9):093001.

\bibitem[Hohmann and Holzmann, 2016]{Hohmann2016}
Hohmann, D. and Holzmann, H. (2016).
\newblock Weighted angle radon transform: Convergence rates and efficient
  estimation.
\newblock {\em Statistica Sinica.}, 26(1):157--175.

\bibitem[Holzmann and Meister, 2020]{HM:20}
Holzmann, H. and Meister, A. (2020).
\newblock Rate-optimal nonparametric estimation for random coefficient
  regression models.
\newblock {\em Bernoulli}, fourthcomming.

\bibitem[Hsiao, 2014]{Hsiao14}
Hsiao, C. (2014).
\newblock {\em Analysis of Panel Data}.
\newblock Cambridge University Press.
\newblock Cambridge Books Online.

\bibitem[Hsiao and Pesaran, 2004]{Hasio04}
Hsiao, C. and Pesaran, M.~H. (2004).
\newblock {Random Coefficient Panel Data Models}.
\newblock CESifo Working Paper Series 1233, CESifo Group Munich.

\bibitem[Ichimura and Thompson, 1998]{Ichimura98}
Ichimura, H. and Thompson, T. (1998).
\newblock Maximum likelihood estimation of a binary choice model with random
  coefficients of unknown distribution.
\newblock {\em Journal of Econometrics}, 86(2):269 -- 295.

\bibitem[Lepski{\u\i}, 1991]{Lepskii:91}
Lepski{\u\i}, O.~V. (1991).
\newblock Asymptotically minimax adaptive estimation. {I}. {U}pper bounds.
  {O}ptimally adaptive estimates.
\newblock {\em Teor. Veroyatnost. i Primenen.}, 36(4):645--659.

\bibitem[Lepski{\u\i}, 1992]{Lepskii:92b}
Lepski{\u\i}, O.~V. (1992).
\newblock On problems of adaptive estimation in white {G}aussian noise.
\newblock In {\em Topics in nonparametric estimation}, volume~12 of {\em Adv.
  Soviet Math.}, pages 87--106. Amer. Math. Soc., Providence, RI.

\bibitem[Masten, 2018]{Masten18}
Masten, M.~A. (2018).
\newblock Random coefficients on endogenous variables in simultaneous equations
  models.
\newblock {\em The Review of Economic Studies}, 85(2):1193--1250.

\bibitem[Masten and Torgovitsky, 2016]{Masten16}
Masten, M.~A. and Torgovitsky, A. (2016).
\newblock Identification of instrumental variable correlated random
  coefficients models.
\newblock {\em The Review of Economics and Statistics}, 98(5):1001--1005.

\bibitem[Math{\'{e}}, 2006]{mathe_2006}
Math{\'{e}}, P. (2006).
\newblock The lepskii principle revisited.
\newblock {\em Inverse Problems}, 22(3):L11--L15.

\bibitem[Nevo, 2001]{Nevo:01}
Nevo, A. (2001).
\newblock Measuring market power in the ready-to-eat cereal industry.
\newblock {\em Econometrica}, 69(2):307--342.

\bibitem[Petrin, 2002]{Petrin:02}
Petrin, A. (2002).
\newblock Quantifying the benefits of new products: The case of the minivan.
\newblock {\em Journal of Political Economy}, 110(4):705--729.

\bibitem[Scherzer et~al., 2009]{Scherzer_etal:09}
Scherzer, O., Grasmair, M., Grossauer, H., Haltmeier, M., and Lenzen, F.
  (2009).
\newblock {\em Variational methods in imaging}, volume 167 of {\em Applied
  Mathematical Sciences}.
\newblock Springer, New York.

\bibitem[Schuster et~al., 2012]{SKHK:12}
Schuster, T., Kaltenbacher, B., Hofmann, B., and Kazimierski, K.~S. (2012).
\newblock {\em Regularization methods in {B}anach spaces}, volume~10 of {\em
  Radon Series on Computational and Applied Mathematics}.
\newblock Walter de Gruyter GmbH \& Co. KG, Berlin.

\bibitem[Tsybakov, 2000]{Tsybakov:00}
Tsybakov, A. (2000).
\newblock On the best rate of adaptive estimation in some inverse problems.
\newblock {\em Comptes Rendus de l'Acad\'emie des Sciences. S\'erie I.
  Math\'ematique}, 330(9):835--840.

\bibitem[Weidling et~al., 2020]{WeSpHo:20}
Weidling, F., Sprung, B., and Hohage, T. (2020).
\newblock Optimal convergence rates for tikhonov regularization in besov
  spaces.
\newblock {\em SIAM Journal on Numerical Analysis}, 58(1):21--47.

\bibitem[Werner, 2018]{Werner:18}
Werner, F. (2018).
\newblock Adaptivity and oracle inequalities in linear statistical inverse
  problems: a (numerical) survey.
\newblock In Hofmann, B., Leitao, A., and Zubelli, J.~P., editors, {\em New
  Trends in Parameter Identification for Mathematical Models}. Birkh\"auser,
  Basel.

\bibitem[Werner and Hohage, 2012]{WH:12}
Werner, F. and Hohage, T. (2012).
\newblock Convergence rates in expectation for {T}ikhonov-type regularization
  of inverse problems with {P}oisson data.
\newblock {\em Inverse Problems}, 28(10):104004, 15.

\end{thebibliography}

\appendix
\section{Proofs}
\subsection{Proofs of Section \ref{sec:convergence}}

\begin{proof}[Proof of Theorem \ref{thm:exist} ]

We apply Proposition 4.2 (i) in \cite{hohage2016inverse} for the existence of a minimizer. We have to check the conditions of the theorem. $T$ is a linear integral operator with square integrable which makes it a compact operator. In addition, $-\bar{\ell}(.)$ and $\mathcal{R}: \mathcal{X} \mapsto \mathbbm{R}$  are both convex and lower semi-continuous. Futhermore, $Tf_\beta \geq 0 \; \text{for all} \; f_\beta \in U$. Hence, the conditions of Proposition 4.2 (i) in \cite{hohage2016inverse} are fulfilled.

The uniqueness of the minimizer follows from the strict convexity of $-\bar{\ell}$ which makes the minimization problem \eqref{eqn:method} strictly convex.

\end{proof}

\begin{proof}[Proof of Theorem \ref{thm:sc}]

Note that the operator $T$ is equivalent to the $d$-dim Radon transformation if  $f_\bTheta$ is bounded away from $0$ on a hemisphere of $\Sp^{d-1}$. Hence, $T^*T$ has eigen functions $\varphi(x) = \exp(i x\cdot k)$ with $k \in \R^d$. The ill-posedness follows from by Theorem 3.1 in \cite{hertle:83}. For all $f \in H^q(\R^d)$, we have $T^*Tf \in H^{q+d-1}(\R^d)$. Using the spectral calculus yields $(T^*T)^\mu f \in H^{q+\mu(d-1)}(\R^d)$. Hence, for $a < r$, $f_\beta$ fulfills the following spectral source condition: There exists $\rho \in H^a(\R^d)$ such that $(T^*T)^\mu \rho = f_\bbeta$ with $\mu = \frac{r-a}{d-1}$. 

As shown in \cite{HoYa:10, Flemming:12, flemming:12b, hohage2016inverse}, this spectral source condition implies the variational source condition 
\[
cD_{\calR}^{f_\beta^{*}}(f,f_\bbeta) \leq \calR(f) - \calR(f_\bbeta) + \lambda(\|f_{S|\bTheta} - Tf\|^2_{L^2})
\]
with some $c>0$, $\calR = \|\cdot\|_{H^a}^2$, and $\lambda(x) = x^\frac{2\mu}{2\mu + 1} =  x^\frac{2(r-a)}{2(r-a)+d-1}$ if $\mu\le 1$. 

We apply the inequality $\|\varphi - \psi\|_{L^2}^2 \le \left(\frac{2}{3}\|\varphi\|_\infty + \frac{4}{3} \|\psi\|_\infty \right) \KL(\varphi;\psi)$ which holds for all nonnegative functions $\varphi, \psi \in L^\infty(\R^d)$ with $\varphi - \psi \in  L^2(\R^d)$, see \cite{BL:91, Dunker2014}. Since $\lambda$ is a power function, we get
\begin{align*}
\lambda(\|f_{S|\bTheta} - Tf\|^2_{L^2}) &\le \lambda\left(\left(\frac{2}{3}\|f_{S|\bTheta}\|_\infty + \frac{4}{3} \| Tf\|_\infty \right) \KL(f_{S|\bTheta}; Tf)\right)\\
& =  \lambda\left(\frac{2}{3}\|f_{S|\bTheta}\|_\infty + \frac{4}{3} \| Tf\|_\infty \right) \lambda \left(\KL(f_{S|\bTheta}; Tf)\right) 
\end{align*}

Hence, there is a constant $c' >0$ such that
\[
c' D_{\calR}^{f_\beta^{*}}(f,f_\bbeta) \leq \calR(f) - \calR(f_\bbeta) + \lambda\left(\KL(f_{S|\bTheta}; Tf)\right)
\]
This proves the theorem.

\end{proof}

\begin{proof}[Proof of Theorem \ref{thm:aprior_conv}]
This is a standard result in variational regularization theory. See for example Theorem 4.3 in \cite{WH:12}.

\end{proof}

\begin{proof}[Proof of Theorem \ref{thm:apriori}]
We apply Theorem \ref{thm:aprior_conv} to the special case of Theorem \ref{thm:sc}. In this case $\lambda(x) = x^\frac{2(r-a)}{2(r-a)+d-1}$ is differentiable. Hence, $-\partial(-\lambda)(x) = \lambda'(x) = \frac{2(r-a)}{2(r-a)+d-1}x^\frac{1-d}{2(r-a)+d-1}$. Replacing $x$ by $n^{-1/2}$ gives the parameter choice. The convergence rate is given by $\lambda(n^{-1/2}) = n^{-\frac{r-a}{2(r-a)+d-1}}$.

\end{proof}

\begin{proof}[Proof of Theorem \ref{thm:lepskii}]
The first part of the theorem is a consequence of Theorem 5.1 in \cite{WH:12}. The second part of the theorem is using the convergence rate from Theorem \ref{thm:apriori} together with the first part of the theorem.
\[
\lambda\left(\frac{\ln(n)}{\sqrt{n}}\right) = \left(\frac{\ln(n)}{\sqrt{n}}\right)^\frac{2(r-a)}{2(r-a)+d-1} = \ln(x)^\frac{2(r-a)}{2(r-a)+d-1}~ n^{-\frac{r-a}{2(r-a)+d-1}}~.
\] 

\end{proof}

\end{document}